\documentclass[11pt,a4]{article}
\pdfoutput=1

\usepackage[T1]{fontenc}
\usepackage{authblk}
\usepackage{amsthm}
\usepackage{amssymb}
\usepackage{amsmath}
\usepackage{amsfonts}
\usepackage{mathtools}
\usepackage{enumerate}
\usepackage{enumitem}
\usepackage{thmtools}
\usepackage{thm-restate}
\usepackage{fullpage}
\usepackage{booktabs}
\usepackage{multirow}
\usepackage{arydshln}
\usepackage{bbm}
\usepackage{changepage} %do wciecia paragrafu poprzez adjustwidth

\usepackage{array}
\newcolumntype{L}[1]{>{\raggedright\let\newline\\\arraybackslash\hspace{0pt}}m{#1}}
\newcolumntype{C}[1]{>{\centering\let\newline\\\arraybackslash\hspace{0pt}}m{#1}}
\newcolumntype{R}[1]{>{\raggedleft\let\newline\\\arraybackslash\hspace{0pt}}m{#1}}

\usepackage{color}
\definecolor{darkgreen}{rgb}{0,0.5,0}
\definecolor{darkblue}{rgb}{0,0,0.8}
\definecolor{darkred}{rgb}{0.8,0,0}
\usepackage{hyperref}
\hypersetup{
   unicode=false,          % non-Latin characters in Acrobat’s bookmarks
   colorlinks=true,        % false: boxed links; true: colored links
   linkcolor=darkred,          % color of internal links (change box color with linkbordercolor)
   citecolor=darkblue,        % color of links to bibliography
   filecolor=magenta,      % color of file links
   urlcolor=black           % color of external links
}

\usepackage{clrscode3e}
\usepackage[noend, noline, ruled, linesnumbered]{algorithm2e}
\SetAlFnt{\normalsize}
\DontPrintSemicolon
\SetKwComment{Comment}{$\triangleright$\ }{}
\SetKwProg{Def}{def}{:}{}
\SetArgSty{textnormal}
\SetKwRepeat{Do}{do}{while}

\newcommand{\bigo}{\mathcal{O}}
\newcommand{\vweight}{\omega}
\newcommand{\pr}{\textup{Pr}}
\newcommand{\target}{v^{*}}
\newcommand{\E}{\mathbb{E}}
\newcommand{\T}{\tau}
\newcommand{\cS}{\mathcal{S}}

\usepackage{mdframed}
\newmdenv[]{aBox}
\newenvironment{algoBox}[1]
    {\begin{center}
     \begin{minipage}{#1\textwidth}
     \begin{aBox}
    }
    {
     \end{aBox}
     \end{minipage}
     \end{center}
    }

\theoremstyle{plain}
\newtheorem{definition}{Definition}[section]
\newtheorem{lemma}[definition]{Lemma}
\newtheorem{theorem}[definition]{Theorem}

\newtheorem{proposition}[definition]{Proposition}

\newtheorem{corollary}[definition]{Corollary}

\newtheorem{algo}[definition]{Algorithm}

\title{Noisy (Binary) Searching: Simple, Fast and Correct\thanks{This work has been partially supported by National Science Centre (Poland) grant number 2018/31/B/ST6/00820.}}

\date{}

\author[1]{\Large Dariusz Dereniowski}
\author[2]{Aleksander~\L{}ukasiewicz}
\author[2]{Przemys\l{}aw~Uzna\'nski}
\affil[1]{\small Gda\'{n}sk~University~of~Technology,~Poland}
\affil[2]{Institute of Computer Science, University~of~Wroc\l{}aw,~Poland}

\begin{document}

\maketitle

\setcounter{page}{0}
\thispagestyle{empty}

\begin{abstract}
This work considers the problem of the noisy binary search in a sorted array.
The noise is modeled by a parameter $p$ that dictates that a comparison can be incorrect with probability $p$, independently of other queries.
We state two types of upper bounds on the number of queries: the worst-case and expected query complexity scenarios.
The bounds improve the ones known to date, i.e., our algorithms require fewer queries.
Additionally, they have simpler statements, and work for the full range of parameters.
All query complexities for the expected query scenarios are tight up to lower order terms.
For the problem where the target prior is uniform over all possible inputs,  we provide an algorithm with expected complexity upperbounded by $(\log_2 n + \log_2 \delta^{-1} + 3)/I(p)$, where $n$ is the domain size, $0\le p < 1/2$ is the noise ratio, and $\delta>0$ is the failure probability, and $I(p)$ is the information gain function.
As a side-effect, we close some correctness issues regarding previous work.
Also, en route, we obtain new and improved query complexities for the search generalized to arbitrary graphs.
This paper continues and improves the lines of research of Burnashev--Zigangirov [Prob. Per. Informatsii, 1974], Ben-Or and Hassidim [FOCS 2008], Gu and Xu [STOC 2023], and Emamjomeh-Zadeh et al. [STOC 2016], Dereniowski et al. [SOSA@SODA 2019].
\end{abstract}

\textbf{Keywords:} Graph Algorithms, Noisy Binary Search, Query Complexity, Reliability

\section{Introduction}

\subsection{Problem statement}

An adaptive search problem for a general \emph{search domain} $\cS$ and an arbitrary \emph{adversary} can be formulated as follows.
The goal is to design an adaptive algorithm, also referred to as a \emph{strategy}, that finds a \emph{target} initially unknown to the algorithm.
Adaptivity means that the subsequent actions of the algorithm depend on the answers already received.
The process is divided into \emph{steps}: in each step the algorithm performs a \emph{query} and receives an \emph{answer}.
Each query-reply pair provides new information to the algorithm: it learns that some part of the search space $S\subseteq\cS$ does not contain the target while its complement does.
From both theoretical and practical viewpoints, it is of interest to develop error-resilient algorithms for such a search process.
This can be modeled, for example, by the presence of a probabilistic noise: each reply can be erroneous with some fixed probability $0<p<\frac{1}{2}$, independently.
The performance of a strategy is measured by the number of performed queries.

\medskip
In this work, we focus on searching with probabilistic noise in two particular types of search domains.
The first is a sorted array (which, in the absence of noise, would lead to the classical binary search problem).
Formally, for a linear order $v_1<\cdots< v_n$ with an unknown position of the target $\target=v_j$, each query selects an element $v_i$, and the algorithm learns from the reply whether  $\target<v_i$ or $\target\geq v_i$.

The second search domain we consider is a simple, undirected graph.
More precisely, for an input graph $G$ and an unknown target vertex $\target$, each query selects some vertex $v$.
The answer either states that $v$ is the target or provides a neighbor $u$ of $v$, that lies on a shortest path from $v$ to $\target$.

Searching through a graph can be viewed as a certain generalization of the former setting, as searching a linear order resembles searching a graph that is a path. However, it is important to note that the two models are not directly comparable, as in a graph search on a path there are three possible replies to a query, whereas in a search through an array, the answers are binary.

\subsection{Overview of our results}

To complete the search process, we need to learn roughly $\log_2 n$ bits of information (the identifier of the target).
We can extract approximately $I(p) = 1 - H(p)$ bits of information from each reply, where $H(p) = -p \log_2 p - (1-p) \log_2 (1-p)$ is the binary entropy function.
% $I(p)$ is the \emph{information function} (also known as \emph{the capacity of Binary Symmetric Channel with crossover probability $p$}).
Therefore we expect the optimal algorithm to use around $\frac{\log_2 n}{I(p)}$ queries.
In an idealized scenario where there always exists a query that perfectly bisects the search space, regardless of the answer, one could achieve this bound.
However, perfect bisection is typically not possible due to the discrete nature of the search space, which causes algorithms to lose some lower-order terms.

\medskip
Our results are summarized in Table~\ref{tab:res}.
We use the following naming convention: if we are interested in the worst-case analysis of the query complexity, we refer to it as the \emph{worst case} setting throughout the paper.
Conversely, if we analyze the expected number of queries made by an algorithm, we call this the \emph{expected query complexity} setting.
We note that in each setting the process is randomized due to answers of the adversary.
Additionally, some of our algorithms use random bits as well.

The binary search algorithms referenced in the theorems below are detailed in Section~\ref{sec:binary-search}: Algorithms~\ref{alg:binary-worst-case} and \ref{alg:bin-search-expected}  correspond to Theorems~\ref{thm:bin:expected} and \ref{thm:bin:worst-case}, respectively.
Interestingly, both algorithms are essentially the same and differ only by the stopping condition.

\begin{theorem} \label{thm:bin:expected}
For any noise parameter $0< p<\frac{1}{2}$ and a confidence threshold $0 < \delta < 1$, there exists a binary search algorithm that after the \textbf{expected} number of
$$\frac{1}{I(p)} \left( \log_2 n + \log_2\delta^{-1} + 3\right)$$
queries finds the target in any linear order correctly with probability at least $1-\delta$, given that \textbf{the target position is chosen uniformly at random}.
\end{theorem}
Using a previously known reduction from adversarial target placement to the uniformly random choice of a target (\cite{Ben-OrH08}, see Lemma~\ref{shiftingtrick}), we automatically obtain the following result for the adversarial version of the problem.

\begin{corollary} \label{cor:bin:expected}
For any noise parameter $0< p<\frac{1}{2}$ and a confidence threshold $0 < \delta < 1$, there exists a binary search algorithm that after the \textbf{expected} number of
$$\frac{1}{I(p)} \left( \log_2 n + \bigo(\log\delta^{-1}) \right)$$
queries finds the target in any linear order correctly with probability at least $1-\delta$.
\end{corollary}

\begin{theorem} \label{thm:bin:worst-case}
For any noise parameter $0< p<\frac{1}{2}$ and a confidence threshold $0 < \delta < 1$, there exists a binary search algorithm for any linear order that after
$$\frac{1}{I(p)} \left( \log_2 n + \bigo(\sqrt{\log n \log\delta^{-1}}) + \bigo(\log\delta^{-1}) \right)$$
queries returns the target correctly with probability at least $1-\delta$.
\end{theorem}

% Taking into account that the leading factor of $1-\delta$ can be here also applied as in \cite{Ben-OrH08}, we correctly obtain the matching bound claimed therein.
% We note that our algorithms can be transformed into strategies of expected query complexity better by roughly a multiplicative factor of $(1-\delta)$ using black-box reductions from \cite{Ben-OrH08} and \cite{gu2023p5aastc}.
% We discuss it further in section \ref{sec:bin-reductions}.
% We also point out that we give the first algorithm for noisy binary search that has the optimal constant in front of the leading term and works for all, possibly non-constant $p$'s in the range $(0, \frac{1}{2})$.

For graph searching we obtain two analogous results (detailed in Section~\ref{sec:graph-search}): Algorithms~\ref{alg:graph-adversarial} and~\ref{alg:graph-expected} correspond to Theorems~\ref{thm:graph:worst-case} and~\ref{thm:graph:expected}, respectively.

\begin{theorem} \label{thm:graph:expected}
For an arbitrary connected graph $G$, a noise parameter $0<p<\frac{1}{2}$ and a confidence threshold $0 < \delta < 1$, there exists an  graph searching algorithm that after the \textbf{expected} number of at most
$$\frac{1}{I(p)} \left( \log_2 n + \bigo(\log \log n) + \bigo(\log\delta^{-1}) \right)$$
queries returns the target correctly with probability at least $1-\delta$.
\end{theorem}

\begin{theorem} \label{thm:graph:worst-case}
For an arbitrary connected graph $G$, a noise parameter $0<p<\frac{1}{2}$ and a confidence threshold $0 < \delta < 1$, there exists an graph searching algorithm that after
$$\frac{1}{I(p)} \left( \log_2 n + \bigo(\sqrt{\log n \log \delta^{-1}}) + \bigo(\log \delta^{-1}) \right)$$
queries returns the target correctly with probability at least $1-\delta$.
\end{theorem}
% Considering the fact that the first and the last terms set the limit on the query complexity, the middle term is the one that we considerably improve over the earlier works \cite{DereniowskiTUW19} and \cite{Emamjomeh-Zadeh:2015aa}.

% In the graph scenario, we argue that just performing the MWU updates technique allows to find the target. In particular, the fact that no perfect bisection is possible in each step is handled on the level of the analysis, and does not have to be addresses in the pseudo-code of the algorithm.
% (For comparison see more complicated algorithms in \cite{Emamjomeh-ZadehK17,Emamjomeh-Zadeh:2015aa}).
% Another contribution lies in a much simpler analysis than in prior works.
% As in the binary search, we managed to have basically the same algorithm behind both theorems with different stopping conditions.

\setlength\dashlinedash{0.2pt}
\setlength\dashlinegap{1.5pt}
\setlength\arrayrulewidth{0.3pt}
\bgroup\def\arraystretch{1.5}

\begin{table}[htb]
\caption{The summary of the results}
\label{tab:res}
\centering
{
\small
\begin{tabular}{p{0.18\linewidth}C{0.37\linewidth}C{0.37\linewidth}}
\toprule
\emph{Setting}                                 & \emph{Binary search}                                                 & \emph{Graph search}                  \\
\midrule
\multirow{2}{\linewidth}{Worst case query complexity:}                     & \multicolumn{2}{c}{$\frac{1}{I(p)}\left(\log_2 n + \bigo(\sqrt{\log n \log \delta^{-1}})+ \bigo(\log \delta^{-1})\right)$} \\
                                                 &  (Thm.~\ref{thm:bin:worst-case})   &   (Thm.~\ref{thm:graph:worst-case})   \\
\cdashline{2-3}
\multirow{2}{\linewidth}{Expected query complexity:}     & $\frac{1}{I(p)}(\log_2 n + \bigo(\log\delta^{-1}))$ & $\frac{1}{I(p)}(\log_2 n + \bigo(\log\log n) + \bigo(\log\delta^{-1}))$
\\
                                                 & (Cor.~\ref{cor:bin:expected}) &  (Thm.~\ref{thm:graph:expected})\\
\bottomrule
\end{tabular}
}
\end{table}

\subsection{Comparison with previous works}
\subsubsection{Upper bounds for noisy binary search}\label{sec:comparison-binary-search}
Table~\ref{tab:upper-binary} provides an overview of the known algorithms for noisy binary search that have complexity close to the optimal $\frac{\log_2 n}{I(p)}$.
We provide a detailed comparison with our results below.
\begin{itemize}
\item Burnashev and Zigangirov \cite{Burnashev1974} studied this problem from an information-theoretic perspective as early as 1974\footnote{Curiously, it appears that until recently, this work has been largely unknown to the algorithmic community, despite the fact that the paper in question has over 100 citations.
There is no mention of \cite{Burnashev1974} in the well-known survey by Pelc \cite{PELC200271}, nor in the subsequent works that we reference.
We suspect that the main reason for this oversight is that, until very recently, the work was only available in Russian.
For the English translation of the algorithm and the proof, see \cite{DBLP:conf/isit/WangGW22}.}, in a setting where the location of the target element is chosen uniformly at random.
We highlight that their query complexity is worse than ours by an additive term of $\log_2 \frac{1-p}{p}$, which tends to infinity when $p \rightarrow 0$.
This behavior is rather unnatural, since when $p=0$, the noise disappears, and the problem reduces to standard binary search.
\item Feige et al. \cite{FeigeRPU94}, in their seminal work, considered several problems in the noisy setting and, in particular, developed an asymptotically optimal algorithm for noisy binary search (with an adversarially placed target). However, their method intrinsically incurs a non-optimal constant in front of $\frac{\log_2 n}{I(p)}$.
\item Later, Ben-Or and Hassidim \cite{Ben-OrH08}, likely unaware of \cite{Burnashev1974}, developed algorithms with an expected query complexity of $\frac{1}{I(p)}(\bigo(\log n) + \bigo(\log \log n) + \bigo(\log \delta^{-1}))$.
However, we claim that their proofs contain two \emph{serious issues}.

Firstly, in the proof of Lemma~2.6 in~\cite{Ben-OrH08}, they consider all the queries made by the algorithm throughout its execution and sort them by their positions.
Then, the number of '$<$' answers in a fixed interval of positions is claimed to follow a binomial distribution.
Notice, however, that while the answers to the particular queries are independent random variables, the \emph{positions} of the queries depend on the answers to the previous queries, and the act of forgetting the order introduces correlation.
To further illustrate this point, we note the rightmost query in their algorithm is guaranteed to have a '$<$' answer, because a '$\geq$' answer would have changed the weights maintained by the algorithm, causing the next query to be asked further to the right, which contradicts the assumption that this query is the rightmost.

Secondly, the final expected number of steps is bounded by the ratio of total information needed to identify the target and the expected information gain per step (without any additional comments).
However, the expected value does not work in this manner directly -- to make this approach effective, one needs to employ additional probabilistic tools.
Our paper uses Wald's identity for this purpose, see \cite{gu2023p5aastc} for an example of the usage of martingales and the Optional Stopping Theorem.
Moreover, the choice of a particular probabilistic theorem and the way this tool is handled may incur additional lower-order terms, making it unclear what the final complexity would be.
\end{itemize}

% \subsection{Correctness issues in \cite{Ben-OrH08}.} \label{sec:issues}
% To motivate our work further, we also point out two issues in the state-of-the art noisy binary search of \cite{Ben-OrH08}.
% Firstly, in that work
% However, this approach in fact requires careful application of probabilistic tools (such as Proposition \ref{thm:wald} in our case) and can incur additional lower-order terms. %(we suspect that careful analysis would reveal additional $\bigo(\log \frac{1-p}{p})$ term in the complexity).

\begin{table}[htb]
  \caption{Upper bounds for noisy binary search.}
  \label{tab:upper-binary}
  \centering
  \begin{tabular}{p{0.1\linewidth}C{0.47\linewidth}R{0.36\linewidth}}
  \toprule
  \emph{Setting} & \emph{Query complexity} & \emph{References and Notes} \\
  \midrule
  \multirow{3}{\linewidth}{Expected, uniform prior}  & $\frac{1}{I(p)}(\log_2 n + \log_2 \delta^{-1} + \log_2 \frac{1-p}{p})$           &  Burnashev and Zigangirov  \cite{Burnashev1974}    \\
         &     $\frac{1}{I(p)}(\log_2 n + \bigo(\log \log n) + \bigo(\log \delta^{-1}))$           &    Ben-Or and Hassidim     \cite{Ben-OrH08} (Correctness issues, see Sec.\ref{sec:comparison-binary-search}).  \\
         &     $\frac{1}{I(p)} \left( \log_2 n + \log_2\delta^{-1} + 3\right)$           &   This work, Thm.~\ref{thm:bin:expected}.    \\
  \cdashline{1-3}
  \multirow{3}{\linewidth}{Expected, adversarial target}  &  $\frac{1}{I(p)}(\bigo(\log n) + \bigo(\log \delta^{-1}))$          &   Feige et al.     \cite{FeigeRPU94}    \\
         &     $\frac{1}{I(p)}(\log_2 n + \bigo(\log \log n) + \bigo(\log \delta^{-1}))$           &    Ben-Or and Hassidim     \cite{Ben-OrH08} (Correctness issues, see Sec.\ref{sec:comparison-binary-search}).  \\
         &     $\frac{1}{I(p)}(\log_2 n + \bigo(\log_2 \delta^{-1}) )$           &   This work, Cor.~\ref{cor:bin:expected}.    \\
  \cdashline{1-3}
  Worst case   & $\frac{1}{I(p)}(\log_2 n + \bigo(\sqrt{\log n \log \delta^{-1}})+ \bigo(\log \delta^{-1}))$          &  This work, Thm.~\ref{thm:bin:worst-case}.    \\

    \bottomrule
  \end{tabular}
\end{table}

\subsubsection{Reductions in complexity for expected length setting}\label{sec:bin-reductions}
Ben-Or and Hassidim in their work \cite{Ben-OrH08} showed a general technique that can transform any of the aforementioned algorithms (regardless if they are in the worst case or expected complexity setting) into an algorithm with the expected query complexity that is better by roughly a multiplicative factor of $(1-\delta)$ at the cost of additive lower order term of order $\frac{\bigo(\log \log n)}{I(p)}$.
Very recently Gu and Xu \cite{gu2023p5aastc} showed how to improve that reduction in order to obtain a better constant in front of $\log \delta^{-1}$.
They plug in our algorithm for noisy binary search (Corollary \ref{cor:bin:expected}) as a black-box in order to get the $(1+o(1))((1-\delta)(\frac{\log_2 n}{I(p)} + \frac{\bigo(\log \log n)}{I(p)}) + \frac{\log_2 1/\delta}{(1-2p)\log \frac{1-p}{p}})$ expected query complexity.

\subsubsection{Upper bounds for noisy graph search}

Known algorithms for noisy graph search are summarized in Table~\ref{tab:upper-graphs}.
The problem for arbitrary graphs was first considered by Emamjomeh-Zadeh et al.\cite{Emamjomeh-Zadeh:2015aa}.
Later on Dereniowski et al. \cite{DereniowskiTUW19} simplified that algorithm and obtained an improved dependence on $\log \delta^{-1}$.
We make another progress in that direction: we further simplify the algorithms and the analysis while simultaneously improving the dependence on $\log \delta^{-1}$ even further.

\begin{table}[htb]
  \caption{Upper bounds for noisy graph search.}
  \label{tab:upper-graphs}
  \centering
  \begin{tabular}{p{0.1\linewidth}C{0.47\linewidth}R{0.36\linewidth}}
  \toprule
  \emph{Setting} & \emph{Query complexity} & \emph{References and Notes} \\
  \midrule
    Expected      &   $\frac{1}{I(p)}(\log_2 n +  \bigo(\log \log n) +  \bigo( \log \delta^{-1}) )$   &   This work, Thm.~\ref{thm:graph:expected}. \\
  \cdashline{1-3}
    \multirow{3}{\linewidth}{Worst case} &  $\frac{1}{I(p)}(\log_2 n +o(\log n)+ \bigo(\log^2 \delta^{-1}))$   &    Emamjomeh-Zadeh et al., 2016 \cite{Emamjomeh-Zadeh:2015aa}  \\
         & $\frac{1}{I(p)}(\log_2 n + \bigo(\sqrt{\log n \log \delta^{-1}}\cdot \log \frac{\log n}{\log \delta^{-1}})+ \bigo(\log \delta^{-1}))$                 &  Dereniowski et al., 2019 \cite{DereniowskiTUW19}   \\
         & $\frac{1}{I(p)}(\log_2 n + \bigo(\sqrt{\log n \log \delta^{-1}})+ \bigo(\log \delta^{-1}))$  & This work, Thm.~\ref{thm:graph:worst-case}.   \\
  \bottomrule
  \end{tabular}
\end{table}

\subsubsection{Lower bounds}
For the expected complexity of noisy binary search, Ben-Or and Hassidim \cite{Ben-OrH08} established the first lower bound of $(1 - \delta)\frac{\log_2 n - 10}{I(p)}$.
Recently Gu and Xu \cite{gu2023p5aastc} improved the lower bound to $(1-o(1))((1-\delta)\frac{\log_2 n}{I(p)} + \frac{\log_2 1/\delta}{(1-2p)\log \frac{1-p}{p}})$.
However it works only for constant noise parameter $p$.
They leave the question of improving the lower bound for an arbitrary $p$ as an open problem.
% We note that the upper bound obtained from combining their techniques with our Corollary \ref{cor:bin:expected} (which we mentioned in Section \ref{sec:bin-reductions}) differs from the lower bound only by an $\frac{\bigo(\log \log n)}{I(p)}$ term.
% We leave closing that gap as a fascinating open problem.

Very recently, Gretta and Price \cite{gretta.price:2024:51stint.colloq.autom.lang.program.icalp2024} obtained a lower bound for the worst case setting of a more general problem (known as \emph{Noisy Binary Search with Monotonic Probabilities}, which was introduced for the first time by \cite{KarpK07}).
For the worst case noisy binary search, their work implies a lower bound of the natural $\frac{\log_2 n}{I(p)}$.
We note that we are not aware of any lower bounds for the lower order terms dependent on $n$ in any of the considered settings.

%\begin{table}[htb]
%  \caption{Lower bounds for noisy binary search in linear orders}
%  \label{tab:lower-binary}
%  \centering
%  \begin{tabular}{p{0.095\linewidth}C{0.3\linewidth}R{0.36\linewidth}}
%  \toprule
%  \emph{Setting} & \emph{Query complexity} & \emph{References and Notes} \\
%  \midrule
%    \multirow{2}{\linewidth}{Expected}   &  $\frac{1 - \delta}{I(p)}(\log n - 10)$  &      Ben-Or and Hassidim, 2008 \cite{Ben-OrH08}      \\
%      & $\frac{\log n}{I(p)} - o(\log n)$   &  Emamjomeh-Zadeh et al., 2020 \cite{emamjomeh-zadeh2020p3icalt}         \\
%  \bottomrule
%  \end{tabular}
%\end{table}

\subsection{Overview of the techniques}
The core building block of our algorithms is Multiplicative Weights Update technique (MWU).
This method has been employed in the past for noisy binary search and related problems \cite{Ben-OrH08, BorgstromK93, DereniowskiTUW19, KarpK07, RivestMKWS80}.
The general outline of the method is as follows: we maintain weights that denote the ''likelihood'' of particular elements being the target and multiplicatively update them according to the answers to subsequent queries.
After a certain number of steps, the algorithm returns the element with the highest weight, as it is the element we deem most likely to be the target.

The typical problem with this approach is that, at some point, we may encounter a situation where there is no good element to query, meaning that no query divides the search space close to the bisection.
This usually occurs when a particular element becomes heavy, and subsequently querying this element yields less and less information.
Previous works tried to different ways to resolve this issue, e.g. by ensuring that a good approximation of the target has been found and calling the algorithm recursively \cite{Ben-OrH08}, introducing a phase with majority voting \cite{Emamjomeh-Zadeh:2015aa}, etc.

We take a different approach by using a specifically tailored measure of progress of our algorithms, which differs from those used in previous works.
For binary search, we define this measure as the total weight minus the weight of the target element.
In the context of graph searching, the measure is slightly different -- we use the total weight minus the weight of the heaviest vertex.
This contrasts with the approach taken in prior studies, such as in \cite{DereniowskiTUW19}, where the total weight itself was utilized.
It is important to note that, in the case of graph search, the identity of the heaviest vertex may change throughout the search process, and at times, it may not even be the target vertex.
However, our analysis guarantees the target will become the heaviest vertex by the end of the search, within the desired probability threshold.

This subtle change proves to be powerful and plays a vital role in all our proofs.
We believe this is the key idea that enabled us to overcome the obstacles that the authors of the previous works might have faced.

Furthermore, in the case of binary search, when selecting which vertex to query, we employ a technique similar to that of \cite{Burnashev1974}.
Specifically, whenever we identify two elements that are closest to bisecting the search space, we randomly choose one of them with appropriate probability.
Our analysis demonstrates that this effectively simulates the ideal subdivision of the search space and ensures the desired progress of our algorithm.

\subsection{Other related work} \label{sec:related-work}
There are many variants of the interactive query games, depending on the structure of queries and the way erroneous replies occur.
The study of such games was initiated by R\'{e}nyi~\cite{Renyi61} and Ulam~\cite{Ulam76}.
A substantial amount of literature deals with a fixed number of errors for arbitrary membership queries or comparison queries; here we refer the reader to surveys~\cite{Deppe2007,PELC200271}.
Among the most successful tools for tackling binary search with errors, is the idea of a volume~\cite{Berlekamp68,RivestMKWS80}, which exploits the combinatorial structure of a possible distribution of errors.
A natural approach of analyzing decision trees has been also successfully applied, see e.g. \cite{FeigeRPU94}.
See \cite{BorgstromK93,Emamjomeh-Zadeh:2015aa} for examples of partitioning strategies into stages, where in each stage the majority of elements is eliminated and only few `problematic' ones remain.
For a different reformulation (and asymptotically optimal results) of the noisy search see \cite{KarpK07}.

Although the adversarial and noisy models are most widely studied, some other ones are also considered.
As an example, we mention the (linearly) bounded error model in which it is guaranteed that the number of errors is a $r$-fraction, $r<\frac{1}{2}$, of (any initial prefix of) the number of queries, see e.g.~\cite{AslamD91,BorgstromK93,DhagatGW92}.
Interestingly, it might be the case that different models are so strongly related that a good bound for one of them provides also tight bounds for the other ones, see e.g. \cite{DereniowskiTUW19}.
We refer the reader to distributional search where an arbitrary target distribution is known to the algorithm a priori \cite{DaganFGM17,DBLP:conf/innovations/DaganFKM21,Shannon48}.
A closely related theory of coding schemes for noisy communication is out of scope of this paper and we only point to some recent works \cite{BravermanEGH16,GellesHKRW16,GellesMS14,Haeupler14,LeungNSTYY18}.

% One possible generalization of this search domain is to consider graph structures, first introduced for trees~\cite{OnakP06} and then extended to general graphs~\cite{Emamjomeh-Zadeh:2015aa}.
The first steps towards generalizing binary search to graph-theoretic setting are works on searching in partially ordered data \cite{Ben-AsherFN99,LaberMP02,LamY01}.
Specifically for the node search that we consider in this work, the first results are due to Onak and Parys for trees~\cite{OnakP06}, where an optimal linear-time algorithm for error-less case was given.
% and the recent works of Emamjomeh-Zadeh et al. for general graphs \cite{Emamjomeh-ZadehK17,Emamjomeh-Zadeh:2015aa}.
%In the former
% It has been shown in~\cite{Emamjomeh-Zadeh:2015aa} how to construct, for the noisy model with an input being an arbitrary graph of order $n$, a strategy of length at most
% $\frac{1}{I(p)}\left(\log_2 n+\bigo(\frac{1}{C}\log_2 n + C^2 \log_2 \delta^{-1})\right)$, where $C=\max\left((\frac{1}{2}-p) \sqrt{\log_2 \log_2 n},1\right)$,
% with the confidence threshold $\delta$.
% The strategy has been simplified and the query complexity further improved in \cite{DereniowskiTUW19} to reach an upper bound on the query complexity:
% $\frac{1}{I(p)}\left(\log_2 n + \widetilde\bigo(\sqrt{\log_2 n \log_2 \delta^{-1}}) + \bigo(\log_2 \delta^{-1})\right)$.
% \footnote{Interestingly, these strategies, when applied to paths, reach the optimal query complexities (up to lower order terms) for linear orders, thus matching the limits of binary search. However one has to have in mind that graph queries applied to linear orders provide \emph{richer} set of replies: $\{<,>,=\}$, compared to $\{<,\ge\}$ in case of a binary search.}

\section{Preliminaries} \label{sec:preliminaries}
% We start with definitions that are common for both search domains.%, followed by those specific to binary search.
% Definitions related strictly to searching in graphs are outlined in appendix \ref{sec:graph-searching}.
Whenever we refer to a \emph{search space}, we mean either an (undirected and unweighted) graph or a linear order.
Consequently, by an \emph{element} of a search space, we refer to a vertex or an integer, respectively.
In the following, let $n$ denote the size of the search space, i.e., either the number of vertices in a graph or the number of integers in a linear order.

Throughout the search process, the strategies will maintain the weights $\vweight(v)$ for the elements $v$ of a search space $V$.
For any $0\leq c\leq 1$, $v$ is \emph{$c$-heavy} if $\vweight(v)/\vweight(V)\geq c$, where for any subset $U\subseteq V$ we write $\vweight(U)=\sum_{u\in U}\vweight(u)$.
$\frac{1}{2}$-heavy elements play a special role and we call them \emph{heavy} for brevity.
The weight of an element $v$ at the end of step $t$ is denoted by $\vweight_t(v)$, with $\vweight_0(v)$ being the initial value.
The initial values are set uniformly by putting $\vweight_0(v)=1$ for each $v$ in our algorithms.

% Having introduced definitions that are common for both graphs and linear orders, we now provide definitions and models that are specific for each type of a search space.

\paragraph*{Noisy binary search definition and model specifics}
In the \emph{noisy binary search} problem, we operate on a linear order $v_1 < \cdots < v_n$.
We are given an element $\target$ and the values $p \in [0, \frac{1}{2}), \delta \in (0, 1)$ as input.
We are promised that there exists some $i \in [n]$ such that $\target = v_i$.
We call $\target$ the \emph{target} element.
We can learn about the search space by asking if $\target < v_j$ for any $j \in [n]$, and receiving an answer that is correct with probability $1-p$, independently for each query.
The goal is to design an algorithm that finds the $i \in [n]$ such that $\target = v_i$, and returns this index correctly with probability at least $1 - \delta$.
We strive to minimize the number of queries performed in the process.

We adopt the following naming convention for query results.
When we ask if $\target \overset{?}{<} v_i$ and receive an affirmative answer (i.e., $\target$ is less than $v_i$), we call it a \emph{yes-answer}.
If the reply indicates $\target$ is greater than or equal to $v_i$ (i.e., a negative answer), we call it a \emph{no-answer}.
An element $v_j$ of a search space is considered \emph{compatible} with the reply to a query $\target \overset{?}{<} v_i$ if and only if:
\begin{itemize}
  \item For a yes-answer (indicating $\target < v_i$), $j<i$.
  \item For a no-answer (indicating $\target \geq v_i$), $j \geq i$.
\end{itemize}

\paragraph*{Noisy graph searching definition and model specifics}
In the \emph{noisy graph searching} problem we are given an unweighted, undirected, simple graph $G$ and the values $p \in [0, \frac{1}{2}), \delta \in (0, 1)$.
We know that one vertex $\target$ of $G$ is marked as the \emph{target}, but we don't know which one is it.

We can query the vertices of $G$, upon querying a vertex $q$ we get one of two possible answers:
\begin{itemize}\setlength\itemsep{0em}
  \item $\target = q$, i.e. the queried vertex is the target.
  We call it a \emph{yes-answer}.
  \item $\target \neq q$, but some neighbor $u$ of $q$ lies on a shortest path from $q$ to $\target$.
  We call it a \emph{no-answer}.
  If there are multiple such neighbors (and hence shortest paths), then we can get an arbitrary one as an answer.
\end{itemize}

In fact, we assume for simplicity that each reply is given as a single vertex $u$.
If $u = q$, then we interpret it as a yes-answer.
If $u \neq q$, then $u$ is a neighbor of $q$ that lies on a shortest path from $q$ to $\target$.
Again, we are interested in the noisy setting, therefore every reply is correct independently with probability $1 - p$.
Observe that if the answer is \emph{incorrect} then it can come in different flavors:
\begin{itemize}\setlength\itemsep{0em}
  \item if $q = \target$, then an incorrect answer is any neighbor $u$ of $q$,
  \item if $q \neq \target$, then an incorrect answer may be either $q$ or any neighbor of $q$ that does not lie on a shortest path from $q$ to $\target$.
\end{itemize}
Clearly, in both cases there might be several possible vertices that constitute an incorrect answer.
Here we assume the strongest possible model where every time the choice among possible incorrect replies is made adversarially and independently for each query.
The goal is, similarly as in noisy binary search, to design an algorithm that finds a target correctly with probability at least $1-\delta$ and minimizes the number of queries.
In fact, in this work we operate in a slightly weaker model of replies (as compared to \cite{DereniowskiTUW19,Emamjomeh-ZadehK17,Emamjomeh-Zadeh:2015aa,OnakP06}) in which an  algorithm receives less information in some cases.
This is done in somewhat artificial way for purely technical reasons, i.e., to simplify several arguments during analysis.
The only change to the model we have just described happens when we query a vertex that is heavy at the moment and a \emph{no-answer} has been received.
More specifically, if a \emph{heavy} $q$ is queried and a \emph{no-answer} is given, the algorithm reads this reply as: the \emph{target is not $q$} (ignoring the direction the target might be).
Observe that this only makes our algorithms stronger, since they operate in a weaker replies model and any algorithmic guarantees for the above model carry over to the generic noisy graph search model.

Similarly to the case of noisy binary search, we say that a vertex $v$ is \emph{compatible} with the reply to the query $q$ if and only if:
\begin{itemize}\setlength\itemsep{0em}
  \item $v = q$ in case of a yes-answer.
  \item The neighbor $u$ given as a no-answer lies on a shortest path from $q$ to $v$ and $q$ was not heavy.
  \item $v \neq q$ in case of a no-answer when $q$ was heavy.
\end{itemize}

\paragraph*{Common mathematical tools and definitions}

We adopt the notation from \cite{DereniowskiTUW19} and denote $\varepsilon=\frac{1}{2}-p$ and $\Gamma = \frac{1-p}{p}$.
These quantities appear frequently throughout the proofs, and this notation helps to make the presentation more concise.

The \emph{information function}, denoted by $I(p)$, appears in all our running times.
It is defined as follows $I(p) = 1 - H(p) = 1 + p \log_2 p + (1-p) \log_2 (1-p)$.
In the analysis of our algorithms we frequently use the following quantitative fact about $I(p)$ and $\log_2 \Gamma$.

\begin{proposition} \label{prop:info-and-gamma-bounds}
  We have $I(p) = \Omega(\varepsilon^2)$ and $(\log_2 \Gamma)^2  p (1-p) = \bigo(\varepsilon^2)$.
\end{proposition}
\begin{proof}
We first observe that by Taylor's series expansion (which can be derived using elementary calculus, see e.g. \cite{4502235}):
$$I(p) = \frac{1}{2\ln 2} \sum_{n=1}^{\infty} \frac{(2\varepsilon)^{2n}}{n(2n-1)} \ge \frac{4\varepsilon^2}{2\ln2}.$$

To show the second bound we compute
$$p(1-p) (\log_2 \Gamma) = (1/2-\varepsilon)(1/2+\varepsilon) \left(\log_2 \frac{1+2\varepsilon}{1-2\varepsilon}\right)^2$$
$$=(1-4\varepsilon^2) \left( \tanh^{-1} 2\varepsilon \right)^2 \frac{1}{ (\ln 2)^2 } \le \frac{4}{(\ln 2)^2} \varepsilon^2$$
where the last step follows by observing that under the substitution $\varepsilon = 1/2 \cdot \tanh \gamma$ it reduces to
$ \gamma^2  \le  \sinh^2 \gamma$ and that inequality follows immediately from the Taylor expansion of $\sinh^2 \gamma$.
\end{proof}

Let $(X_m)_{m\in \mathbb{N}}$ be a sequence of i.i.d random variables.
We say that a random variable $T$ is a \emph{stopping time} (with respect to $(X_m)_{m\in \mathbb{N}}$) if $\mathbbm{1}_{\{T \leq m\}}$ is a function of $X_1, X_2, \ldots, X_m$ for every $m$.

We will use the following version of the Wald's identity.

\begin{proposition}[Wald's Identity \cite{wald1945ams}]\label{thm:wald}
  Let $(X_m)_{m\in \mathbb{N}}$ be i.i.d with finite mean, and $T$ be a stopping time with $\E[T] < \infty$. Then $\E[X_1 + \cdots + X_T] = \E[X_1]\E[T]$.
  \end{proposition}

Let us also recall a basic version of a Hoeffding bound, which we use in our calculations of query complexities in the worst case setting.

% \begin{proposition}[Hoeffding bound \cite{hoeffding1963jasa}]\label{thm:hoeffding}
% 	Let $X_1,\ldots,X_m$ be $m$ independent random variables with support in $[a,b]$. Define $X := \sum_{i=1}^{n} X_i$. Then, for every $t > 0$,
% 		$$\pr(X - \E[X] \ge t) \leq \exp\{-\frac{2t^2}{m(b-a)^2}\}, $$
%     $$\pr(X - \E[X] \le -t) \leq \exp\{-\frac{2t^2}{m(b-a)^2}\}. $$
% \end{proposition}

\paragraph*{Multiplicative Weights Update Method}
The core building block of our strategies for both binary and graph search is a standard Multiplicative Weights Update (MWU) technique.
Below we formally define the version of MWU that we use in our algorithms (Algorithm~\ref{alg:MWU}).
\begin{algoBox}{0.85}
\begin{algo} \label{alg:MWU} \textup{(MWU updates.)}

\medskip
In a step $t+1$, for each element $v$ of the search space do:

\hspace*{15pt}if $v$ is compatible with the answer, then $\vweight_{t+1}(v) \leftarrow \vweight_t(v) \cdot 2(1-p)$,

\hspace*{15pt}if $v$ is not compatible with the answer, then $\vweight_{t+1}(v) \leftarrow \vweight_t(v) \cdot 2p$.
\end{algo}
\end{algoBox}

Directly from the statement of our MWU method we obtain the following bound on the weight of the target.
This bound applies to both binary and graph search, as the analysis is based solely on the number of erroneous replies and the fact that the target is always compatible with a correct answer.
\begin{lemma} \label{lem:target-bound}
If $\target$ is the target, then after $\T$ queries, with probability at least $1-\delta$ it holds
$$\vweight_{\T}(\target) \ge  \Gamma^{- \sqrt{2 p(1-p) \T \ln \delta^{-1}}} 2^{I(p) \T}.$$
\end{lemma}
\begin{proof}
After $\T$ queries with at most $\ell$ erroneous replies, the weight of the target satisfies:
$$\vweight_{\T}(\target) \ge (2p)^{\ell} (2(1-p))^{\T-\ell} = \Gamma^{p\T - \ell} 2^{I(p) \T}.$$
Denote $a = \sqrt{2 p(1-p) \T  \ln \delta^{-1}}$.
Then by Chernoff-Hoeffding bound \cite{hoeffding1963jasa}, with probability at most $\delta$ there is $\ell - p\T \ge a$.
Thus, after $\T$ queries, with probability at least $1-\delta$ the weight of the target satisfies $\vweight_{\T}(\target) \ge \Gamma^{-a} 2^{I(p)\T}.$
\end{proof}

\paragraph*{Uniform prior for binary search}
One can assume that the distribution of the target element in noisy binary search is \emph{a priori} uniform by using a shifting trick described by Ben-Or and Hassidim \cite{Ben-OrH08}.
We formally state it as a lemma below.

\begin{lemma}[c.f. \cite{Ben-OrH08}]
\label{shiftingtrick}
  Assume that the target element in noisy binary search problem was chosen adversarially.
  One can reduce that problem to the setting where the target element is chosen uniformly at random using $\bigo (\frac{\log \delta^{-1}}{I(p)})$ queries.
\end{lemma}

\section{Binary Search Algorithm} \label{sec:binary-search}

Each query performs the MWU updates using Algorithm~\ref{alg:MWU}.
The element to be queried is selected as follows: let $k$ be such that $\sum_{i=1}^{k-1} \vweight(v_i) \le \vweight(V)/2$ and $\sum_{i=1}^k\vweight(v_i) \ge \vweight(V)/2$.
Since the queries $v_k$ and $v_{k+1}$ are the closest possible to equi-division of the total weight, the algorithm chooses one of those with appropriate probability (cf. Algorithm~\ref{alg:bin-search-basic}.)
\begin{algoBox}{1.0}
\begin{algo} \label{alg:bin-search-basic} \textup{(Query selection procedure.)}

\medskip
In step $\T$: let $k$ be such that $\sum_{i=1}^{k-1} \vweight_{\T}(v_i) \le \frac{\vweight_{\T}(V)}{2} \le \sum_{i=1}^{k} \vweight_{\T}(v_i)$.

Then, query $v_k$ with probability $\frac{1}{2\vweight_{\T}(v_k)}(\sum_{i=1}^{k} \vweight_{\T}(v_i) - \sum_{i=k+1}^{n} \vweight_{\T}(v_i))$, and  otherwise query $v_{k+1}$.
\end{algo}
\end{algoBox}
In order to turn Algorithm~\ref{alg:bin-search-basic} into a particular strategy, we will provide a stopping condition for each model.
We start by determining the expected weight preservation during the search.

\begin{lemma} \label{lem:nontargetdrop}
$\E[\vweight_{\T+1}(V \setminus \{\target\})\ |\ \vweight_{\T}(V \setminus \{\target\})] \le \vweight_{\T}(V \setminus \{\target\}).$
\end{lemma}
\begin{proof}
We consider three cases, and show that this bound holds in each of those independently.
Denote $A = \sum_{i=1}^{k-1} \vweight_{\T}(v_i)$, $B = \vweight_{\T}(v_k)$ and $C = \sum_{i=k+1}^{n} \vweight_{\T}(v_i)$.
Denote the probability of querying $v_k$ as $\alpha = \frac{A+B-C}{2B}$, and the probability of querying $v_{k+1}$ as $\beta =  \frac{C+B-A}{2B}$.

\noindent
Case 1: $\target = v_k$.
\begin{eqnarray*}
\E[\vweight_{\T+1}(V \setminus \{\target\})\ |\ \vweight_{\T}(V \setminus \{\target\})] & = & \alpha[2p^2C + 2(1-p)^2 C+2p(1-p)A+2p(1-p)A] \\
 & & + \beta[2p^2A + 2(1-p)^2 A+2p(1-p)C+2p(1-p)C] \\
 & = & \alpha[(1+4\varepsilon^2) C+(1-4\varepsilon^2)A] + \beta[(1+4\varepsilon^2) A+(1-4\varepsilon^2)C].
\end{eqnarray*}
Using the definition of $\alpha$ and $\beta$, we obtain
%$$=\frac{A-C}{2B}4\varepsilon^2(C-A) + \frac{C-A}{2B} 4\varepsilon^2(A-C) + \frac{B}{2B}[(1+4\varepsilon^2) C+(1-4\varepsilon^2)A+(1+4\varepsilon^2) A+(1-4\varepsilon^2)C]$$
\begin{eqnarray*}
\E[\vweight_{\T+1}(V \setminus \{\target\})\ |\ \vweight_{\T}(V \setminus \{\target\})] & = & -4\varepsilon^2 \frac{(A-C)^2}{B}+ (A+C) \\
 & = & -4\varepsilon^2 \frac{(A-C)^2}{B} +  \vweight_{\T}(V \setminus \{\target\}).
\end{eqnarray*}

\noindent
Case 2: $\target < v_k$. In this case,
\begin{eqnarray*}
\E[\vweight_{\T+1}(V \setminus \{\target\})\ |\ \vweight_{\T}(V \setminus \{\target\})] & = & (2p^2 + 2(1-p)^2)(A-\vweight_{\T}(\target)) \\
 & & + [\alpha 4p(1-p)+ \beta(2p^2+2(1-p)^2)]B+4p(1-p)C.
\end{eqnarray*}
Denote $p_1 = p^2+(1-p)^2$ and $p_2 = 2p(1-p)$. Observe $p_1+p_2 = 1$ and $p_1 \ge \frac{1}{2} \ge p_2$. Then,
\begin{eqnarray*}
\E[\vweight_{\T+1}(V \setminus \{\target\})\ |\ \vweight_{\T}(V \setminus \{\target\})] & = & 2p_1 A + (C+B-A)p_1 + (A+B-C)p_2 + 2p_2C - 2p_1\vweight_{\T}(\target) \\
 &=& A+B+C - 2p_1 \vweight_{\T}(\target)\\
 &\le & \vweight_{\T}(V \setminus \{\target\}).
\end{eqnarray*}

\noindent
Case 3: $\target > v_{k+1}$ is symmetric to case 2.
\end{proof}

\subsection{Proof of Theorem~\ref{thm:bin:expected} (The expected strategy length)} \label{sec:bin-search:las-vegas}

\begin{algoBox}{0.95}
\begin{algo} \label{alg:bin-search-expected} \textup{(The expected strategy length for binary search.)}

\medskip
Initialization: $\vweight_0(v_i)\leftarrow1$ for each $v_i\in V$.\\
Execute Algorithm~\ref{alg:bin-search-basic} until in some step $\T$ it holds $\frac{\vweight_{\T}(v_j)}{\vweight_{\T}(V)} \ge 1-\delta$ for some $v_j$.\\
Return $v_j$.
\end{algo}
\end{algoBox}

To show correctness of Algorithm~\ref{alg:bin-search-expected}, we need to observe that in the case of binary search, our MWU updates are, in fact Bayesian updates, that is, the normalized weights follow a posterior distribution conditioned on the replies seen so far.
We state this as a lemma below.

\begin{lemma}\label{lem:bayesian-posterior}
  After any step $\T$ of Algorithm~\ref{alg:bin-search-expected} we have for every $v_i \in V$
  $$
    \Pr[\target = v_i\ |\ \T \text{ observed replies}] = \frac{\vweight_{\T}(v_i)}{\vweight(V)}.
  $$
\end{lemma}
\begin{proof}
  The proof is by induction on $\T$.
  The base case is trivial, because we assign the weights uniformly in the initialization step.
  The inductive step follows immediately from our definition of MWU updates (Algorithm~\ref{alg:MWU}) and the Bayes' rule.
\end{proof}

The correctness of Algorithm~\ref{alg:bin-search-expected} follows directly from Lemma~\ref{lem:bayesian-posterior} and the stopping condition.
To prove Theorem~\ref{thm:bin:expected} it remains to analyze the expected length of the strategy.

\begin{lemma}
Algorithm~\ref{alg:bin-search-expected} terminates after the expected number of at most $\frac{1}{I(p)} (\log_2 n + \log_2 \frac{1}{\delta} + 3)$ steps.
\end{lemma}
\begin{proof}
We measure the progress at any given step by a random variable $\zeta_t = \log_2 \vweight_t(\target)$.
If the answer in step $t+1$ is erroneous, then $\zeta_{t+1} = \zeta_t + 1 + \log_2 p$ and otherwise $\zeta_{t+1} = \zeta_t + 1 + \log_2(1-p)$.

For the sake of bounding the number of steps of the algorithm, we consider it running indefinitely.
Let $Q$ be the smallest integer such that $\frac{\vweight_Q(\target)}{\vweight_Q(V\setminus \{\target\})} \ge \frac{1-\delta}{\delta}$, that is $\target$ is $(1-\delta)$-heavy in round $Q$.
Obviously $Q$ upper bounds the strategy length.
By the definition, $Q$ is a stopping time.

First, let us show that $\E[Q]$ is finite.
To this end, let $\xi_t = \log_2 \frac{\vweight_t(\target)}{\vweight_t(V\setminus \{\target\})}$ and $X_t = \xi_t - \xi_{t-1}$ for $t~\ge~1$.
Observe that $X_t$'s are i.i.d and $\E[X_t] = p(-\log_2 \Gamma) + (1-p)\log_2 \Gamma = (1-2p)\log_2 \Gamma > 0$.
Therefore, using Lemma~\ref{lem:random-walk} for sequence $X_t$, with $\ell = -\log \Gamma$, $r = \log \Gamma$ and $T = \log_2 \frac{1-\delta}{\delta}-\log_2 \frac{\vweight_0(\target)}{\vweight_0(V\setminus \{\target\})}$, we indeed obtain $\E[Q] < \infty$.

Let $Q_i$ for any positive integer $i$ be smallest value such that $\vweight_{Q_i}(\target) \ge \frac{n}{\delta}\cdot 2^i$ (for completeness of notation, we define $Q_0 = 0$).
Consider an event $\{Q > Q_i\}$. It means that in round $Q_i$ the target $\target$ is not yet $(1-\delta)$-heavy.
Hence, $\vweight_{Q_i}(V \setminus \{\target\}) > \frac{\delta}{1-\delta}\vweight_{Q_i}(\target) \ge \delta \vweight_{Q_i}(\target) \ge  n \cdot 2^i$.
But we know from Lemma~\ref{lem:nontargetdrop} that $\E[\vweight_{Q_i}(V \setminus \{\target\})] \leq \E[\vweight_0(V \setminus \{\target\})] \leq n$.
Using Markov's inequality we conclude that $\pr[Q > Q_i] \leq \pr[\vweight_{Q_i}(V \setminus \{\target\}) > n \cdot 2^i] \leq 2^{-i}$.

Additionally, since $\vweight_{Q_i-1}(\target) < \frac{n}{\delta}\cdot 2^i$, there is $\vweight_{Q_i}(\target) < 2\frac{n}{\delta}\cdot 2^i$.
We can then bound
\begin{eqnarray*}
\E[\zeta_Q] &<& \sum_{i=1}^{\infty} \Pr[Q_{i-1} < Q \le Q_i] \log_2 (2 \cdot \frac{n}{\delta}2^i) \\
 &=& \log_2 (2 \cdot \frac{n}{\delta}) + \sum_{i=1}^{\infty} \Pr[Q_{i-1} < Q \le Q_i] \cdot i\\
 &=& \log_2(2 \frac{n}{\delta}) + \sum_{i=1}^{\infty} \Pr[Q>Q_{i-1}]\\
 &\le & 1 + \log_2 n + \log_2 \frac{1}{\delta} + \sum_{i=0}^{\infty} 2^{-i}\\
 &\le & \log_2 n + \log_2 \frac{1}{\delta} + 3.
\end{eqnarray*}

Let $Y_t = \zeta_t - \zeta_{t-1}$.
Obviously, $Y_i$'s are independent.
We have already established that $Q$ is a stopping time and that $\E[Q] < \infty$.
This means we can employ the Wald's identity (Proposition \ref{thm:wald}) to obtain (using $\zeta_0 = 0$)
$\E[\zeta_Q]= \E[Y_1 + \cdots + Y_Q] = \E[Q]I(p).$
Therefore,
$$\log_2 n + \log_2 \frac{1}{\delta} + 3 \ge \E[\zeta_Q] = \E[Q] I(p).\qedhere$$
\end{proof}

\subsection{Proof of Theorem~\ref{thm:bin:worst-case} (Worst-case strategy length)} \label{sec:bin-search:worst-case}
Take $Q$ to be the smallest positive integer for which

\begin{equation}\label{eq:bin-adv}
  I(p) Q > \log_2 n + \log_2 \frac{2}{\delta} + \sqrt{2p(1-p) Q \ln \frac{2}{\delta}} \log_2 \Gamma.
\end{equation}

The $Q$ gives our strategy length (see Algorithm~\ref{alg:binary-worst-case}).
To prove Theorem~\ref{thm:bin:worst-case} we bound the strategy length and the failure probability (see Lemma~\ref{lem:bin-adv} below). The algorithm essentially remains the same except for the stop condition (cf. Algorithm~\ref{alg:binary-worst-case}).
\begin{algoBox}{0.75}
\begin{algo} \label{alg:binary-worst-case} \textup{(Worst-case strategy length for binary search.)}

\medskip
Initialization: $\vweight_0(v_i) \leftarrow 1$ for each element $v_i$.\\
Execute Algorithm~\ref{alg:bin-search-basic} for exactly $Q$ steps with $Q$ as in~\eqref{eq:bin-adv}.\\
Return the heaviest element.
\end{algo}
\end{algoBox}

\begin{lemma} \label{lem:bin-adv}
For any $0 < \delta < 1$, Algorithm~\ref{alg:binary-worst-case} finds the target correctly with probability at least $1-\delta$ in $\frac{1}{I(p)} \left( \log_2 n + \bigo(\log\delta^{-1}) + \bigo(\sqrt{\log n \log\delta^{-1}}) \right)$ steps.
\end{lemma}
\begin{proof}

  Firstly, observe that solving~\eqref{eq:bin-adv} for $Q$ using Lemma~\ref{quadratic} with parameters $a = I(p)$, $b=\log_2 n + \log_2 \frac{2}{\delta}$ and $c=2p(1-p)(\log_2 \Gamma)^2 \ln \frac{2}{\delta}$ yields
$$
Q = \frac{1}{I(p)}\left(\log_2 n +  \log_2 \frac{2}{\delta} +  \bigo\left(\ln \frac{2}{\delta} + \sqrt{\log_2 n +  \log_2 \frac{2}{\delta} } \sqrt{\ln \frac{2}{\delta}}\right)\right)
$$
where we have used, by Proposition \ref{prop:info-and-gamma-bounds}, that $\frac{p(1-p)(\log_2 \Gamma)^2}{I(p)} = \bigo(1)$. The above equation can be simplified to
$$
  Q = \frac{1}{I(p)} \left( \log_2 n + \bigo(\log \delta^{-1}) + \bigo(\sqrt{\log \delta^{-1} \log n}) \right).
$$

It remains to prove correctness.
By Lemma~\ref{lem:target-bound}, with probability at least $1-\delta/2$ we have
\begin{equation}\label{eq2137}
\vweight_{Q}(\target) \ge  \Gamma^{- \sqrt{2p(1-p) Q \ln 2/\delta}} 2^{I(p) Q}.
\end{equation}

From Lemma~\ref{lem:nontargetdrop} we also get $\E[\vweight_Q(V \setminus \{\target\})] \leq \E[\vweight_0(V \setminus \{\target\})] \leq n$.
Therefore, by Markov's inequality with probability $1-\delta/2$ we have
\begin{equation} \label{eq2}
  \vweight_Q(V \setminus \{\target\}) \le \frac{2n}{\delta}.
\end{equation}

It remains to observe that our definition of $Q$ (Equation~\eqref{eq:bin-adv}) is equivalent to
$$
  \Gamma^{- \sqrt{2p(1-p) Q \ln 2/\delta}} 2^{I(p) Q} > \frac{2n}{\delta}.
$$

Thus, by the union bound applied to \eqref{eq2137} and \eqref{eq2}, with probability at least $1-\delta$ we get $\vweight_Q(\target) > \vweight_Q(V \setminus \{\target\})$.
Then, $\target$ is the heaviest element and Algorithm~\ref{alg:binary-worst-case} returns it.
\end{proof}

\section{Graph Searching Algorithm} \label{sec:graph-search}
Denoting by $d(u,v)$ the graph distance between $u$ and $v$, i.e., the length of the shortest path between these vertices, a \emph{median} of the graph is a vertex
$$q = \arg \min_{v \in V} \sum_{u \in V} d(u,v) \cdot \vweight(u).$$
For a query $v$ and a reply $u$, let $C(v, u) = \{x \in V | u $
lies on some shortest path from $v \text{ to } x \}$ for $v \neq u$, and $C(v, v) = \{v\}$.
We note a fundamental bisection property of a median:
\begin{lemma}[c.f. \cite{Emamjomeh-Zadeh:2015aa} Lemma 4] \label{lem:median-bisection}
If $q$ is a median, then $\max_{u \in N(q)} \vweight(C(q, u))  \le \vweight(V)/2$.
\end{lemma}

\begin{proof}%[Proof of Lemma~\ref{lem:median-bisection}]
  Denote for brevity $\Phi(x)=\sum_{v \in V} d(x,v) \cdot \vweight(v)$ for any $x\in V$.
  Suppose towards the contradiction that $C(q, u) > \vweight(V)/2$ for some $u \in N(q)$.
  Observe that $\Phi(u) \leq \Phi(q) - \vweight(C(q, u)) + \vweight(V \setminus C(q, u))$ since by moving from $q$ to $u$ we get closer to all vertices in $C(q, u)$.
  But $\Phi(q) - \vweight(C(q, u)) + \vweight(V \setminus C(q, u)) = \Phi(q) + \vweight(V) - 2\vweight(C(q, u)) < \Phi(q)$ by our assumption, hence $\Phi(u) < \Phi(v)$, which yields a contradiction.
  \end{proof}

We now analyze how the weights behave when in each step a median is queried and the MWU updates are made.
This analysis is common for both graph searching algorithms given later.
Essentially we prove that, in an amortized way, the total weight (with the heaviest vertex excluded) remains the same in each step.
In absence of heavy vertices we use Lemma~\ref{lem21}.
Lemmas~\ref{lem22} and~\ref{lem23} refer to an interval of queries to the same heavy vertex $x$.
If the interval ends (cf. Lemma~\ref{lem22}), then the desired weight drop can be claimed at its end.
For this, informally speaking, the crucial fact is that $x$ received many no-answers during this interval.
If a strategy is at a step that is within such interval, then Lemma~\ref{lem23} is used to bound the total weight with the weight of $x$ excluded.
Hence, at any point of the strategy the weight behaves appropriately, as summarized in Lemma~\ref{lem:xT}.

\begin{lemma}[see also \cite{DereniowskiTUW19, Emamjomeh-Zadeh:2015aa}]
\label{lem21}
If in a step $t$ there is no heavy vertex, then
$\vweight_{t+1}(V) \le  \vweight_t(V).$
\end{lemma}
\begin{proof}
Let $q$ be a query and $u$ an answer in step $t$.
Note that if there is no heavy vertex, then $C(q, u)$ is the set of vertices compatible with the reply.
If $q \neq u$ then $\vweight_{t}(C(q, u))\leq\vweight_{t}(V)/2$ by Lemma \ref{lem:median-bisection} and in case $q = u$ we have $C(q, u)=\{q\}$ and thus the same bound holds.
Then in both cases,
$\vweight_{t+1}(V) =  2(1-p)\cdot\vweight_{t}(C(q, u))+2p\cdot\vweight_{t}(V \setminus C(q, u))   \leq \vweight_{t}(V)$.
\end{proof}

\begin{lemma}[see also \cite{DereniowskiTUW19}]
\label{lem22}
Consider an interval $I=\{\T,\T+1,\ldots,\T+k-1\}$ of $k$ queries such that some $x$ is heavy in each query in $I$ and is not heavy after the last query in the sequence. Then
$\vweight_{\T+k}(V) \le  \vweight_{\T}(V).$
\end{lemma}
\begin{proof}%[Proof of Lemma~\ref{lem22}]
  First note that in each query in the interval $I$ the queried vertex is $x$.
  Consider any two queries $i$ and $j$ in $I$ such that they receive different replies.
  The contribution of these two queries is that together they scale down each weight multiplicatively by $2p\cdot2(1-p)\leq1$.
  Also, for a single no-answer in a query $i \in I$ we get
  $\vweight_{i+1}(V)=2p\vweight_i(x)+2(1-p)\vweight_i(V\setminus\{x\}) \leq \vweight_i(V)$
  because $\vweight_i(x)\geq\vweight_i(V\setminus\{x\})$ for the heavy vertex $x$.
  By assumption, the number of no-answers is at least the number of yes-answers in $I$.
  Thus, the overall weight drop is as claimed in the lemma.
  \end{proof}

 \begin{lemma}
 \label{lem23}
 Consider an interval $I=\{\T,\T+1,\ldots,\T+k-1\}$ of $k$ queries such that some $x$ is heavy in each query in $I$, and $x$ remains heavy after the last query in $I$.
 Then $$\vweight_{\T+k}(V \setminus \{x\}) \le  \vweight_{\T}(V).$$
 \end{lemma}
 \begin{proof}
 Recall that in each query in the interval $I$, the queried vertex is $x$.
 Assume that there were $a$ yes-answers in $I$ and $b$ no-answers, with $a+b=k$.
 If $a \ge b$, then $\vweight_{\T+k}(V \setminus \{x\}) = (2p)^a (2(1-p))^b \vweight_{\T}(V \setminus \{x\}) \le \vweight_{\T}(V \setminus \{x\}) \le \vweight_{\T}(V)$.
 If $a < b$, then we bound as follows:
 $\vweight_{\T+k}(V \setminus \{x\}) \le \vweight_{\T+k}(x) = (2p)^b (2(1-p))^a \vweight_{\T}(x) \le \vweight_{\T}(x) \le \vweight_{\T}(V)$.
 \end{proof}

The bound in the next lemma immediately follows from Lemmas~\ref{lem21},~\ref{lem22} and~\ref{lem23}.
We say that an element $v$ is \emph{heaviest} if $\vweight(v)\geq\vweight(u)$ for each $u\in V$.
For each step $i$, we denote by $x_i$ a heaviest vertex at this step, breaking ties arbitrarily.

\begin{lemma} \label{lem:xT}
$\vweight_{\T}(V \setminus \{x_{\T}\}) \le \vweight_0(V)=n.$
\end{lemma}
\begin{proof}
We consider the first $\tau$ queries and observe that they can be partitioned into a disjoint union of maximal intervals in which either there is a heavy vertex present (in the whole interval) or there is no heavy vertex (in the whole interval).
We apply Lemma~\ref{lem21} for intervals with no heavy vertex and Lemmas~\ref{lem22},~\ref{lem23} otherwise (note that Lemma~\ref{lem23} can be applied only to the last interval.
The latter happens only when there exists a heavy vertex after we perform all $\tau$ queries).
\end{proof}

\subsection{Proof of Theorem~\ref{thm:graph:worst-case} (Worst-case strategy length)}
In this section we prove Theorem~\ref{thm:graph:worst-case}.
Take $Q$ to be the smallest positive integer for which
\begin{equation}\label{eq:Q-definition}
I(p)Q \ge \log_2 n + \sqrt{2p(1-p)Q \ln \delta^{-1}} \log_2 \Gamma.
\end{equation}

The $Q$ gives our strategy length (see Algorithm~\ref{alg:graph-adversarial}).
To prove Theorem~\ref{thm:graph:worst-case} we bound the strategy length and the failure probability (see Lemma~\ref{lem:graph:worst-case} below).

\begin{algoBox}{1.0}
\begin{algo} \label{alg:graph-adversarial} \textup{(Worst-case strategy length for graph search.)}

\medskip
Initialization: $\vweight_0(v)=1$ for each $v\in V$.\\
In each step: query the median and perform the MWU updates (Algorithm~\ref{alg:MWU}).\\
Stop condition: do exactly $Q$ queries with $Q$ defined by \eqref{eq:Q-definition} and return the heaviest vertex.
\end{algo}
\end{algoBox}

\begin{lemma} \label{lem:graph:worst-case}
For any $0 < \delta < 1$, Algorithm~\ref{alg:graph-adversarial} finds the target correctly with probability at least $1-\delta$ in $\frac{1}{I(p)} \left( \log_2 n + \bigo(\log\delta^{-1}) + \bigo(\sqrt{\log n \log\delta^{-1}}) \right)$ steps.
\end{lemma}
\begin{proof}
%Let $x_Q$ be the heaviest vertex at the end of the strategy, i.e., in step $Q$ defined in~\eqref{eq:Q-worst-case}.

The proof is very similar to that of Lemma \ref{lem:bin-adv} (worst-case noisy binary search).
It is actually simpler, thanks to the fact that Lemma \ref{lem:xT} gives even stronger bound than Lemma \ref{lem:nontargetdrop}.

Using Lemma~\ref{quadratic} we solve \eqref{eq:Q-definition} for $Q$.
We bound the result further with Proposition \ref{prop:info-and-gamma-bounds}:
\begin{equation} \label{eq:Q-worst-case}
Q = \frac{\log_2 n + \bigo(\sqrt{\log n \log \delta^{-1}}) + \bigo(\log \delta^{-1})}{I(p)}.
\end{equation}
%It remains to show correctness.
By Lemma~\ref{lem:target-bound} and the definition of $Q$ in~\eqref{eq:Q-definition} it holds with probability $1-\delta$
$$\log_2 \vweight_Q(\target) \ge - \sqrt{2p(1-p) Q\ln \delta^{-1}} \log_2 \Gamma + I(p)Q \ge \log_2 n \ge \log_2 \vweight_Q(V \setminus \{x_Q\}),$$
where the last inequality is due to Lemma~\ref{lem:xT}.
Since the weights are non-negative at all times, the only way for this to happen is to have $\target = x_Q$, that is the target being found correctly.
\end{proof}

\subsection{Proof of Theorem~\ref{thm:graph:expected} (The expected strategy length)}

The solution for this case (given in Algorithm~\ref{alg:graph-expected}) paraphrases Algorithm~\ref{alg:bin-search-basic} except for the proper adjustement of the confidence threshold.
\begin{algoBox}{1.0}
\begin{algo} \label{alg:graph-expected} \textup{(The expected strategy length for graph search.)}

\medskip

%Let $\delta' = \Theta(\frac{\delta^2 }{(\log n+\log 1/\delta)^2}).$
Let $\delta' = c(\delta^2 \cdot (\log n+\log 1/\delta)^{-2})$ for small enough constant $c>0$.

Initialization: $\vweight_0(v)=1$ for each $v\in V$.\\
In each step: query the median and perform the MWU updates (Algorithm~\ref{alg:MWU}).\\
Stop condition: if for any $v$ in some step $\T$ it holds $\frac{\vweight_{\T}(v)}{\vweight_{\T}(V)} \ge 1-\delta'$, then return $v$.
\end{algo}
\end{algoBox}

\begin{lemma} \label{lem:graph-lv-distr}
Algorithm~\ref{alg:graph-expected} stops after the expected number of at most $\frac{1}{I(p)} (\log_2 n +  \log_2 \frac{1}{\delta'} + 1)$ steps.
\end{lemma}

\begin{proof}
We argue that within the promised expected number of steps, target $\target$ reaches the threshold weight. This clearly upperbounds the runtime.
We measure the progress at any given moment by a random variable $\zeta_t = \log_2 \vweight_t(\target)$.
Observe that if the reply is erroneous in a step $t+1$, then $\zeta_{t+1} = \zeta_t + 1 + \log_2 p$, and if it is correct, then $\zeta_{t+1} = \zeta_t + 1 + \log_2 (1-p)$.

For the sake of bounding the number of steps of the algorithm, we assume it is simulated indefinitely.
Let $Q$ be the smallest integer such that $\zeta_Q \ge \log_2  n + \log_2 \frac{1-\delta'}{\delta'}$.

By Lemma \ref{lem:xT} we have that $\zeta_Q = \log_2 \vweight_Q(\target) \le \log_2 \frac{\vweight_Q(\target)}{\vweight(V \setminus \{x_Q\})/n}$, thus $\frac{\vweight_Q(\target)}{\vweight(V \setminus \{x_Q\})} \ge \frac{1-\delta'}{\delta'} > 1$ since w.l.o.g. $\delta' < 1/2$.
But if for any $t$ there is $\frac{\vweight_t(\target)}{\vweight_t(V \setminus \{x_t\})} > 1$, then $x_t = \target$, since $\vweight_t(\target) > \vweight_t(V \setminus \{x_t\})$ implies $\target \not\in V \setminus \{x_t\}$.
Thus we deduce that $x_Q = \target$.
Additionally, from $\frac{\vweight_Q(\target)}{\vweight(V \setminus \{\target\})} \ge \frac{1-\delta'}{\delta'}$ we get that $\target$ is $(1-\delta')$-heavy, hence $Q$ bounds the strategy length.

From $\zeta_{t+1} \in \{\zeta_t + 1 + \log_2 p, \zeta_t + 1 + \log_2 (1-p)\}$ and the minimality of $Q$ we deduce $\zeta_Q \le \log_2 n + \log_2 \frac{1-\delta'}{\delta'} + 1 + \log_2(1-p) \le \log_2 n + \log_2 \frac1{\delta'} + 1$.
In particular
\begin{equation} \label{eq:zetaQ}
\E[\zeta_Q] \leq \log_2 n + \log_2 \frac1{\delta'} + 1.
\end{equation}
Let $X_t = \zeta_t - \zeta_{t-1}$ and observe that $\E[X_t] = p(1 + \log_2 p) + (1-p)(1 +\log_2 (1-p)) = I(p) > 0$.
Also, $X_i$'s are independent and $Q$ is a stopping time.
Finally, we have $\E[Q] < \infty$ from Lemma~\ref{lem:random-walk} \footnote{By plugging in $\ell = 1 + \log_2 p$, $r = 1 + \log_2 (1-p)$ and $T = \log_2 \frac{1-\delta'}{\delta'} + \log_2 n$.}.
Therefore, we meet all conditions of the Wald's identity (Proposition \ref{thm:wald}) and we get (since $\zeta_0 = 0$)
$\E[\zeta_Q]= \E[X_1 + \cdots + X_Q] = \E[Q]I(p).$
Thus, by~\eqref{eq:zetaQ} we have
$1 + \log_2 \frac1{\delta'} + \log_2 n \ge \E[\zeta_Q] = \E[Q] I(p),$
from which the claim follows.
\end{proof}

\begin{lemma} \label{lem:graph:lv-adversarial}
Algorithm~\ref{alg:graph-expected} finds the target correctly with probability at least $1-\delta$.
\end{lemma}

\begin{proof}%[Proof of Lemma~\ref{lem:graph:lv-adversarial}]
  We first show correctness.
  Denote by $A \le \frac{\log \frac{1-\delta'}{\delta'}}{\log \frac{1-p}{p}}+1$ the number of yes-answers required to go from a vertex being $1/2$-heavy to being $(1-\delta')$-heavy. For now assume that $A \ge 2$, we will deal with the other case later.
  For a non-target vertex $u$ to be declared by the algorithm as the target, it has to observe a suffix of the strategy being a random walk on a 1-dimensional discrete grid $[0, \ldots, A]$ and transition probabilities $p$ for $i \to i+1$ and $1-p$ for $i \to i-1$.
  We consider a random walk starting at position $A/2$ and ending when reaching either $0$ or $A$ and call it a \emph{subphase} (w.l.o.g. assume that $A$ is even).
  Any execution of the algorithm can be partitioned into maximal in terms of containment, disjoint subphases.
  Each subphase starts when one particular heavy vertex $v$ receives $A/2$ more yes-answers than no-answers within the interval in which $v$ is heavy.
  Then, a subphase ends when either the algorithm declares $v$ to be the target or $v$ stops being heavy.
  By the standard analysis of the gamblers ruin problem, each subphase (where the heavy vertex is not the target) has failure probability $\delta''=  \frac{1}{1+\left(\frac{1-p}{p}\right)^{A/2}} \le \frac{1}{1+\sqrt{\frac{1-\delta'}{\delta'}}} = O(\sqrt{\delta'})$.
  Let us denote by a random variable $D$ the number of subphases in the execution of the algorithm. Let $F_i$ be the  length of $i$-th subphase.
  By the standard analysis of the gamblers ruin problem,
  $$\E[F_i] = \frac{A/2}{1-2p} - \frac{A}{1-2p}\frac{1}{1 + (\frac{1-p}{p})^{A/2}} \ge \frac{A/2}{1-2p}\left(1 - \frac{2}{1 + \sqrt{\frac{1-\delta'}{\delta'}}}\right) = \Omega\left(\frac{1}{\varepsilon^2}\right),$$
  where the asymptotic holds since w.l.o.g. $\delta' < 1/3$, and also since if $\varepsilon < 1/3$, then $A = \Omega(1/\varepsilon)$, and otherwise $A \ge 2 = \Omega(1/\varepsilon)$.
  Let $F = F_1 + \cdots + F_D$ be the total length of all subphases.
  Observe that $D$ is a stopping time, hence we have $\E[F] = \E[D] \cdot \Omega(\frac{1}{\varepsilon^2})$ by Proposition \ref{thm:wald}.
  By Lemma~\ref{lem:graph-lv-distr}, $\E[Q]=\bigo(\varepsilon^{-2}(\log n+\log\delta'^{-1}))$ holds for the strategy length $Q$.
  Since $F \le Q$, $\E[D] = \bigo(\log n+\log 1/\delta') = \bigo(\log n + \log 1/\delta)$.

  By application of the union bound, the error probability for the whole procedure is bounded by $\delta'' \E[D] \le \delta$ for appropriately chosen constant in the definition of $\delta'$.

  We now deal with case of $A\le1$. This requires $p < \delta'$, and $\varepsilon > 1/3$ (since if $\varepsilon<1/3$, appropriate choice of constant in $\delta'$ enforces $A\ge2$) and so the expected strategy length is $\E[Q] = O(\log n + \log 1/\delta)$. By the union bound, algorithm receives a single erroneous response with probability at most $p \E[Q] \le \delta' \E[Q] = O(\delta^2/(\log n + \log 1/\delta)) \le \delta$.
  \end{proof}

%\newpage % te dwie linie do usuniecia w camera ready
%\setcounter{page}{1}

%\bibliographystyle{plain}
\bibliography{bib}

\appendix

\section{Delegated Proofs}

\begin{lemma}
\label{quadratic}
The  solution to $ax = b + \sqrt{cx}$ is of the form $x = \frac{1}{a}\left(b + \bigo\left(\frac{c}{a}+\sqrt{b}\cdot\sqrt{\frac{c}{a}}\right)\right)$.
\end{lemma}
\begin{proof}
Solving the quadratic equation $a^2x^2 + b^2 - 2abx - cx = 0$, we get:
$$\Delta = (2ab+c)^2 - 4a^2b^2 = c(c + 4ab)$$
$$x = \frac{2ab+c + \sqrt{c^2 + 4abc}}{2a^2} = \frac{b}{a} + \bigo\left(\frac{c+\sqrt{abc}}{a^2}\right)$$
\end{proof}

\begin{lemma}\label{lem:random-walk}
Let $(X_i)_{i \in \mathbb{N_+}}$ be a sequence of i.i.d random variables with $\pr(X_i = \ell) = p$ and $\pr(X_i = r) = (1-p)$ for some $\ell, r \in \mathbb{R}$ and $0 < p < \frac{1}{2}$.
Let us fix $T > 0$ and let $Q = \inf \{m : \sum \limits_{i=1}^m X_i \geq T\}$.
If $\E[X_i] > 0$, then $\E[Q] < \infty$.
\end{lemma}

\begin{proof}
  Let $\zeta_m = \sum \limits_{i=1}^m X_i$.
  If $Q > m$, then in particular $\zeta_m \le T$, hence
  \begin{equation} \label{eq:random-walk-1}
    \pr(Q > m) \le \pr(\zeta_m \le T)
  \end{equation}
  for any $m \in \mathbb{N_+}$.

  Let $\mu = \E[X_i]$.
  Obviously, $\E[\zeta_m] = m\E[X_i] = m\mu$.
  Now, let us define $N = \lceil \frac{T}{\mu} \rceil$.
  For any $m > N$ we have
  \begin{equation} \label{eq:random-walk-2}
    \zeta_m \le T \iff \zeta_m - m\mu \le -(m\mu - T)
  \end{equation}
  and $m\mu - T > 0$.
  Using Hoeffding bound \cite{hoeffding1963jasa} we get
  \begin{equation} \label{eq:random-walk-3}
    \pr(\zeta_m - m\mu \le -(m\mu - T)) \le \exp\{\frac{-2(m\mu - T)^2}{m(\ell + r)^2}\} \le \exp\{\frac{-2m\mu^2}{(\ell + r)^2} + \frac{4\mu T}{(\ell + r)^2}\} = C \beta^{m}
  \end{equation}
  with $C = e^\frac{4\mu T}{(\ell + r)^2}$ and $\beta = e^{\frac{-2\mu^2}{(\ell + r)^2}}$.
  Observe that $0 < \beta < 1$.

  Putting together equations \eqref{eq:random-walk-1}, \eqref{eq:random-walk-2} and \eqref{eq:random-walk-3} we get
  $$ \E[Q] = \sum \limits_{m=0}^\infty \pr(Q > m) = \sum \limits_{m \le N} \pr(Q > m) + \sum \limits_{m > N} \pr(Q > m) \le \sum \limits_{m \le N} \pr(Q > m) +  C\sum \limits_{m > N} \beta^m < \infty.$$
\end{proof}

\end{document}